\documentclass[11pt,twoside]{article}
\usepackage[letterpaper,hmargin=1in,vmargin=1in]{geometry}
\usepackage{times} 

\usepackage{amssymb,amsmath,amsthm} 
\usepackage{url}
\usepackage{enumitem}
\usepackage{calc}
\usepackage{subimages}
\setfigdir{.}

\newtheorem{theorem}{Theorem}
\newtheorem{definition}{Definition}
\newtheorem{corollary}{Corollary}
\newtheorem{lemma}{Lemma}

\newcounter{problemcounter}
\newenvironment{problem}[4][]
{
	\medskip\noindent {\textbf{Problem
	\arabic{problemcounter}} (\textsc{#1}).\nopagebreak

	{\begin{tabular}{p{1.9cm}p{10.2cm}}	
		\textsc{Instance}: & #2 \\
		\textsc{Parameter}: & #3 \\
		\textsc{Question}: & #4 \\
	\end{tabular}}
	\stepcounter{problemcounter}}
	\medskip
}

\newcommand{\mathsymbol}[2]{\newcommand{#1}{\ensuremath{#2}}}
\newcommand{\mathcommand}[3]{\newcommand{#1}[1]{\ensuremath{#2{##1}#3}}}

\newcommand{\funname}[1]{\ifmmode \mathsf{#1}\else\textsf{#1}\fi}

\mathsymbol{\Mmatch}{M}
\mathsymbol{\Mmatchr}{\tilde{\Mmatch}}
\mathcommand{\Mmatchset}{\mathcal{M}(}{)}
\mathcommand{\Mmatchsetr}{\tilde{\mathcal{M}}(}{)}

\begin{document}


\title{Parameterized Complexity of Discrete Morse Theory}

\author{
             Benjamin A.~Burton\thanks{
	School of Mathematics and Physics, 
	The University of Queensland,
	Brisbane, Australia,
	\texttt{bab@maths.uq.edu.au}}
\and
            Thomas Lewiner\thanks{
	Department of Mathematics,
	Pontif{\'i}cia Universidade Cat{\'o}lica,
	Rio de Janeiro, Brazil,
	\texttt{thomas@lewiner.org}}
\and 
             Jo{\~a}o Paix{\~a}o\thanks{
	Department of Mathematics,
	Pontif{\'i}cia Universidade Cat{\'o}lica,
	Rio de Janeiro, Brazil,
	\texttt{jpaixao@mat.puc-rio.br}}
\and
             Jonathan Spreer\thanks{
	School of Mathematics and Physics,
	The University of Queensland,
	Brisbane, Australia,
	\texttt{j.spreer@uq.edu.au}}
}
\date{}

\maketitle

\begin{abstract}
Optimal Morse matchings reveal essential structures of cell complexes which lead to powerful tools to study discrete geometrical objects, in particular discrete $3$-manifolds. However, such matchings are known to be NP-hard to compute on $3$-manifolds, through a reduction to the erasability problem.

Here, we refine the study of the complexity of problems related to discrete Morse theory in terms of parameterized complexity. On the one hand we prove that the erasability problem is $W[P]$-complete on the natural parameter. On the other hand we propose an algorithm for computing optimal Morse matchings on triangulations of $3$-manifolds which is fixed-parameter tractable in the treewidth of the bipartite graph representing the adjacency of the $1$- and $2$-simplexes. This algorithm also shows fixed parameter tractability for problems such as erasability and maximum alternating cycle-free matching.
We further show that these results are also true when the treewidth of the dual graph of the triangulated $3$-manifold is bounded. Finally, we investigate the respective treewidths of simplicial and generalized triangulations of $3$-manifolds.
\end{abstract}



\noindent
{\bf Keywords}: discrete Morse theory, parameterized complexity, fixed parameter tractability, treewidth, $W[P]$-completeness, computational topology, collapsibility, alternating cycle-free matching





\section{Introduction}\label{sec:introduction}
Classical Morse theory~\cite{morse1931smoothmorse} relates the topology of a manifold to the critical points of scalar functions defined on it, providing efficient tools to understand essential structures on manifolds. Forman~\cite{forman1998morse} recently extended this theory to arbitrary cell complexes. In this discrete version of Morse theory, alternating cycle-free matchings in the Hasse diagram of the cell complex, so-called \emph{Morse matchings}, play the role of smooth functions on the manifold~\cite{forman1998morse,chari2000decomposition}. For example, similarly to the smooth case~\cite{reeb1946points}, a closed manifold admitting a Morse matching with only two unmatched (critical) elements is a sphere~\cite{forman1998morse}.
The construction of specific Morse matchings has proven to be a powerful tool to understand topological~\cite{forman1998morse,joswig2004computing,joswig2006computing,lewiner2003optimal,lewiner2003toward}, combinatorial~\cite{chari2000decomposition,jonsson2005graph,lange2004topocombinatorics} and geometrical~\cite{gyulassy2008visualization,lewiner2005geometric,Robins2011,reininghaus2012computationalmorse} structures of discrete objects.

Morse matchings that minimize the number of critical elements are known as {\em optimal matchings}~\cite{lewiner2003optimal}. Together with their number and type of critical elements, these are topological (more precisely simple homotopy) invariants of the cell complex, just like in the case of the sphere described above. Hence, computing optimal matchings can be used as a purely combinatorial technique in computational topology~\cite{Dey1999comptopo}. Moreover, optimal Morse matchings are useful in practical applications such as volume encoding~\cite{lewiner2004apps,szymczak1999grow}, or homology and persistence computation~\cite{lewiner2005geometric,guenther12}.

However, constructing optimal matchings is known to be NP-hard on general $2$-complexes and on $3$-manifolds~\cite{joswig2004computing,joswig2006computing,lewiner2003optimal}. 
This result follows from a reduction to this problem from the closely related \emph{erasability problem}: how many faces must be deleted from a $2$-dimensional simplicial complex before it can be completely erased, where in each erasing step only {\em external triangles}, i.e.\ triangles with an edge not lying in the boundary of any other triangle of the complex, can be removed~\cite{eugeciouglu1996computationally}?
Despite this hardness result, large classes of inputs -- for which worst case running times suggest the problem is intractable -- allow the construction of optimal Morse matchings in a reasonable amount of time using simple heuristics~\cite{lewiner2003toward}. Such behavior suggests that, while the problem is hard to solve for some instances, it might be much easier to solve for instances which occur in practice. As a consequence, this motivates us to ask what {\em parameter} of a problem instance is responsible for the intrinsic hardness of the optimal matching problem. 

In this article, we study the complexity of Morse type problems in terms of parameterized complexity. Following Downey and Fellows~\cite{downey1994parameterized}, an NP-complete problem is called {\em fixed-parameter tractable} (FPT) with respect to a {\em parameter} $k \in \mathbb{N}$, if for every input with parameter less or equal to $k$, the problem can be solved in $O(f(k)\cdot n^{O(1)})$ time, where $f$ is an arbitrary function independent of the problem size $n$.
For NP-complete but fixed-parameter tractable problems, we can look for classes of inputs for which fast algorithms exist, and identify which aspects of the problem make it difficult to solve. Note that the significance of an FPT result strongly depends on whether the parameter is (i) small for large classes of interesting problem instances and (ii) easy to compute. 

In order to also classify fixed-parameter intractable NP-complete problems, Downey and Fellows~\cite{downey1994parameterized} propose a family of complexity classes called the {\em $W$-hierarchy}:
$FPT \subseteq W[1] \subseteq  W[2] \subseteq \cdots \subseteq W[P] \subseteq XP$.
The base problems in each class of the $W$-hierarchy are versions of satisfiability problems with increasing logical depth as parameter. 
Class $W[P]$ contains the satisfiability problems with unbounded logical depth. 
The rightmost complexity class $XP$ of the $W$-hierarchy contains all problems which can be solved in $O(n^k)$ time where $k$ is the parameter of the problem.

Here, we use the notion of the $W$-hierarchy in a geometric setting. More precisely, we determine the hardness of Morse type problems using the mathematically rigid framework of the $W$-hierarchy. Our first main result shows that the erasability problem is $W[P]$-complete (Theorems \ref{thm:inWP} and \ref{thm:WPhard}), where the parameter is the {\em natural parameter} -- the number of cells that have to be removed. In other words, we prove that the erasability problem is fixed-parameter intractable in this parameter. From a discrete Morse theory point of view, this reflects the intuition that reaching optimality in Morse matchings requires a global (at least topological) context. In this way, we also show that the $W$-hierarchy as a purely complexity theoretical tool can be used in a very natural way to answer questions in the field of computational topology.
Although there are many results about the computational complexity of topological problems~\cite{agol06-knotgenus,burton2012complexity,eugeciouglu1996computationally,malgouyres2008determining,tancer2008d}, to the authors' knowledge, erasability is the first purely geometric problem shown to be $W[P]$-complete.

Our second main result refines the observation that simple heuristics allow us to compute optimal matchings efficiently. For general $2$-complexes (and $3$-manifolds), the problem reduces directly to finding a maximal alternating cycle-free matching on a spine, i.e., a bipartite graph representing the $1$- and $2$-cell adjacencies~\cite{ayala2012perfect, joswig2006computing, lewiner2004apps} (Lemma~\ref{lem:morse_manifold}). To solve this problem, we propose an explicit algorithm for computing maximal alternating cycle-free matchings which is fixed-parameter tractable in the treewidth of this bipartite graph (Theorem \ref{thm:treewidth}). Furthermore, we show that finding optimal Morse matchings on triangulated $3$-manifolds is also fixed-parameter tractable in the treewidth of the dual graph of the triangulation (Theorem~\ref{thm:fpg}), which is a common parameter when working with triangulated $3$-manifolds~\cite{burton2012complexity}.

Finally, we use the classification of simplicial and generalized triangulations of $3$-manifolds to investigate the ``typical'' treewidth of the respective graphs for relevant instances of Morse type problems. In this way, we give further information on the relevance of the fixed parameter results. The experiments show that the average treewidths of the respective graphs of simplicial triangulations of $3$-manifolds are particularly small in the case of generalized triangulations. Furthermore, experimental data suggest a much more restrictive connection between the treewidth of the dual graph and the spine of triangulated $3$-manifolds than the one stated in Theorem \ref{thm:fpg}.


\section{Preliminaries}\label{sec:prelims}

\paragraph*{Triangulations}
Throughout this paper we mostly consider simplicial complexes of dimensions $2$ and $3$, although most of our results hold for more general combinatorial structures. All $2$-dimensional simplicial complexes we consider are (i) pure, i.e., all maximal simplexes are triangles ($2$-simplexes) and (ii) {\em strongly connected}, i.e., each pair of triangles is connected by a path of triangles such that any two consecutive triangles are joined by an edge ($1$-simplex). All $3$-dimensional simplicial complexes we consider are triangulations of closed $3$-manifolds, that is, simplicial complexes whose underlying topological space is a closed $3$-manifold. In particular every $3$-manifold can be represented in this way~\cite{moise1952}. We will refer to these objects as {\em simplicial triangulations of $3$-manifolds}.

In Section \ref{sec:experimental} we briefly concentrate on a slightly more general notion of a \emph{generalized triangulation of a $3$-manifold}, which is a collection of tetrahedra all of whose faces are affinely identified or ``glued together'' such that the underlying topological space is a $3$-manifold. Generalized triangulations use far fewer tetrahedra than simplicial complexes, which makes them important in computational $3$-manifold topology (where many algorithms are exponential time). Every simplicial triangulation is a generalized triangulation, and the second barycentric subdivision of a generalized triangulation is a simplicial triangulation~\cite{moise1952}, hence both objects are closely related.

For the remainder of this article, we will often consider $2$-dimensional simplicial complexes as part of a simplicial triangulation of a $3$-manifold.

\paragraph*{Erasability of simplicial complexes}
Let $\Delta$ be a $2$-di\-men\-sion\-al sim\-pli\-cial complex.
A triangle $t \in \Delta$ is called \textit{external} if $t$ has at least one edge which is not in the boundary of any other triangle in $\Delta$; otherwise $t$ is called \textit{internal}. Given a $2$-dimensional simplicial complex $\Delta$ and a triangle $t \in \Delta$, the $2$-dimensional simplicial complex obtained by removing (or {\em erasing}) $t$ from $\Delta$ is denoted by $\Delta \setminus t$. In addition, if $\Delta'$ is obtained from $\Delta$ by iteratively erasing triangles such that in each step the erased triangle is external in the respective complex, we will write $\Delta \rightsquigarrow \Delta'$. We say that the complex $\Delta$ is \textit{erasable} if $\Delta \rightsquigarrow \emptyset$, where in this context $\emptyset$ denotes a complex with no triangle.
Finally, for every $2$-dimensional simplicial complex $\Delta$ we define $\operatorname{er}(\Delta)$ to be the size of the smallest subset $\Delta_0$ of triangles of $\Delta$ such that $\Delta \setminus \Delta_0 \rightsquigarrow \emptyset$. The elements of $\Delta_0$ are called {\em critical} triangles and hence $\operatorname{er}(\Delta)$ is sometimes also referred to as the {\em minimum number of critical triangles} of $\Delta$.
Determining $\operatorname{er}(\Delta)$ is known as the erasability problem~\cite{eugeciouglu1996computationally}.

\begin{problem}[Erasability]
{A $2$-dimensional simplicial complex $\Delta$.}
{A non-negative integer $k$.}
{Is $\operatorname{er}(\Delta)\le k$?}
\end{problem}

\paragraph*{Hasse diagram and spine}
Given a simplicial complex $\Delta$, one defines its {\em Hasse diagram} $H$ to be a directed graph in which the set of nodes of $H$ is the set of simplexes of $\Delta$, and an arc goes from $\tau$ to $\sigma$ if and only if $\sigma$ is contained in $\tau$ and $ \dim (\sigma)+1 = \dim (\tau)$. Let $H_i \subseteq H$ be the bipartite subgraph spanned by all nodes of $H$ corresponding to $i$- and $i+1$-dimensional simplexes. In particular, $H_1$ describes the adjacency between the $2$-simplexes and $1$-simplexes of $\Delta$, and will be called the {\em spine} of the simplicial complex $\Delta$. The spine of a simplicial complex will be one of the main objects of study in this work. 

\paragraph*{Matchings}
By a {\em matching of a graph} $G = (N,A)$ we mean a subset of arcs $M \subset A$ such that every node of $N$ is contained in at most one arc in $M$. Arcs in $M$ are called {\em matched arcs} and the nodes of the matched arcs are called {\em matched nodes}. Nodes and arcs which are not matched are referred to as {\em unmatched}. By the {\em size} of a matching $M$ we mean the number of matched arcs. A matching $M$ is called a {\em maximum matching} of a graph $G$ if there is no matching with a larger size than the size of $M$.

\paragraph*{Morse matchings}
Let $H$ be the Hasse diagram of a simplicial complex $\Delta$ and $M$ be a matching on $H$. Let $H(M)$ be the directed graph obtained from the Hasse diagram by reversing the direction of each arc of the matching $M$. If $H(M)$ is a directed acyclic graph, i.e., $H(M)$ does not contain directed cycles, then $M$ is a {\em Morse matching}~\cite{chari2000decomposition}. Furthermore, the number $c_i$ of unmatched vertices representing $i$-simplexes of $\Delta$ is called the number of \emph{critical $i$-dimensional simplexes} and the sum $c(M) = \sum_i c_i$ is said to be the {\em total number of critical simplexes}.

The motivation to find optimal Morse matchings is given by the following fundamental theorem of discrete Morse theory due to Forman which deals with simple homotopy~\cite{Cohen1973}.

\begin{theorem}[\cite{forman1998morse}]\label{forman_theorem}
Let $M$ be a Morse matching on a simplicial complex $\Delta$. Then $\Delta$ is simple homotopy equivalent to a $CW$-complex with exactly one $d$-cell for each critical $d$-simplex of $M$.
\end{theorem}

In other words, a Morse matching with the smallest number of critical simplexes gives us the most compact and succinct topological representation up to homotopy. For more information about the basic facts of Morse theory we refer the reader to Forman's original work~\cite{forman1998morse}.
This motivates a fundamental problem in discrete Morse theory,
\emph{optimal Morse matching}, as a decision problem in the following form.

\begin{problem}[Morse Matching]
{A simplicial complex $\Delta$.}
{A non-negative integer $k$.}
{Is there a Morse matching $M$ with $c(M) \leq k$?}
\end{problem}
\noindent
Note that \textsc{Erasability} can be restated as a version of \textsc{Morse Matching} where only the number of unmatched $2$-simplexes (that is, $c_2 (M)$) is counted~\cite{lewiner2003optimal}.

\paragraph*{Complexity of Morse matchings}
The compelexity of computing optimal Morse matchings is linear on 1-complexes (graphs)~\cite{forman1998morse} and
$2$-complexes that are manifolds~\cite{lewiner2003optimal}.
Joswig and Pfetsch~\cite{joswig2006computing} prove that if you can solve \textsc{Erasability} in the spine of a $2$-simplicial complex in polynomial time, then you can solve \textsc{Morse matching} in the \textit{entire} complex in polynomial time. The proof technique easily extends to $3$-manifolds, leading to the following lemma which has been mentioned in previous works \cite{lewiner2004apps, ayala2012perfect}.
\begin{lemma}
	\label{lem:morse_manifold} 
	Let $M$ be a Morse matching on a triangulated $3$-manifold $\Delta$. Then we can compute a Morse matching $M'$ in polynomial time which has exactly one critical $0$-simplex and one critical $3$-simplex, such that $c(M') \leq c(M)$.
\end{lemma}

In other words, answering \textsc{Erasability} on the spine is the only difficult part when solving \textsc{Morse Matching} on a $3$-manifold. In Section \ref{sec:treewidth} we show that if a spine has bounded treewidth, then we can solve \textsc{Erasability} in linear time.
Lemma \ref{lem:morse_manifold} therefore generalizes this result to \textsc{Morse Matching} on $3$-manifolds.
%
%


\section{W[P]-Completeness of the erasability problem}
\label{sec:WpCompleteness}

In order to prove that \textsc{Erasability} is $W[P]$-complete in the natural parameter, we first have to take a closer look at what has to be considered when proving hardness results with respect to a particular parameter.

\begin{definition}[Parameterized reduction]\label{def:parameterized_reduction} 
	A parameterized problem $L$ reduces to a parameterized problem $L'$, denoted by $L \leq_{FPT} L'$, if we can transform an instance $(x, k)$ of $L$ into an instance $(x',g(k))$ of $L'$ in time $f(k)|x|^{O(1)}$ (where $f$ and $g$ are arbitrary functions), such that $(x, k)$ is a yes-instance of $L$ if and only if $(x',g(k))$ is a yes-instance of $L'$.
\end{definition}

As an example, E{\u{g}}ecio{\u{g}}lu and Gonzalez~\cite{eugeciouglu1996computationally} reduce \textsc{Set Cover} to \textsc{Erasability} to show that \textsc{Erasability} is NP-complete. Since their reduction approach turns out to be a parameterized reduction, these results can be restated in the language of parameterized complexity as follows.

\begin{corollary}
	\textsc{Set Cover} $\leq_{FPT}$ \textsc{Erasability}, therefore \textsc{Erasability} is $W[2]$-hard.
\end{corollary}

This shows that, if the parameter $k$ is simultaneously bounded in both
problems, \textsc{Erasability} is \textit{at least as hard as}
\textsc{Set Cover}. In this section we will determine exactly how much
harder \textsc{Erasability} is than \textsc{Set Cover}, which is
$W[2]$-complete. Namely, we will show that \textsc{Erasability} is
$W[P]$-complete in the natural parameter $k$. This will be done by i)
using a $W[P]$-complete problem as an oracle to solve an arbitrary
instance of \textsc{Erasability} (Theorem \ref{thm:inWP}, which shows
that \textsc{Erasability} is in $W[P]$), and ii) reducing an arbitrary instance of a suitable problem which is known to be $W[P]$-complete to an instance of \textsc{Erasability} (Theorem \ref{thm:WPhard}, which shows that \textsc{Erasability} is $W[P]$-hard).

\medskip
There are only a few problems described in the literature which are known to be $W[P]$-complete \cite[p.~473]{downey1999parameterized}. Amongst these problems, the following is suitable for our purposes.

\begin{problem}[Minimum Axiom Set]
{A finite set $S$ of {\em sentences}, and an {\em implication relation} $R$ consisting of pairs $(U,s)$ where $U \subseteq S$ and $s \in S$.}
{A positive integer $k$.}
{Is there a set $S_0 \subseteq S$ (called an {\em axiom set}) with $|S_0| \leq k$ and a positive integer $n$, for which $S_n = S$, where we define $S_i$, $1 \leq i \leq n$, to consist of exactly those $s \in S$ for which either $s \in S_{i-1}$ or there exists a set $U \subseteq S_{i-1}$ such that $(U, s) \in R$?}
\end{problem}
\begin{theorem}[\cite{downey1994parameterized}]
	\textsc{Minimum Axiom Set} is $W[P]$-complete.
\end{theorem}

In this paper, we show that, preserving the natural parameter $k$, \textsc{Minimum Axiom Set} is both at least and at most as hard as \textsc{Erasability}.
\begin{theorem}\label{thm:inWP}
	\textsc{Erasability} $\leq_{FPT}$ \textsc{Minimum Axiom Set}, therefore \textsc{Erasability} is in $W[P]$.
\end{theorem}
Theorem \ref{thm:inWP} shows that \textsc{Erasability} is at most as hard as the hardest problems in $W[P]$. Please refer to the full version of this paper for a detailed proof of Theorem~\ref{thm:inWP}. 

In order to show that it is in fact amongst the hardest problems in this class we first need to build some gadgets.
\begin{definition}[Gadgets for the hardness proof of \textsc{Erasability}]\label{def:implication_gadget}
	Let $(S,R,k)$ be an instance of \textsc{Minimum Axiom Set}. 

	Let $s \in S$ be a sentence. By an {\em $s$-gadget} or {\em sentence gadget} we mean a triangulated $2$-dimensional sphere with $2n + m$ punctures as shown in \figref{gadget_part}, where $m$ is the number of relations $(U,s) \in R$ and $n$ is the number of relations $(U,u) \in R$ such that $s \in U$.	
 
	Let $(U,s) \in R$ be a relation. A {\em $(U,s)$-gadget} or {\em implication gadget} is a collection of $|U|+1$ sentence gadgets for each sentence of $U\cup\{s\}$ together with $2|U|$ nested tubes as shown in \figref{gadget} such that (i) two tubes are attached to two punctures of the $u$-gadget for each $u \in U$ and (ii) all $2|U|$ boundary components at the other side of the tubes are identified at a single puncture of the $s$-gadget.
\end{definition}

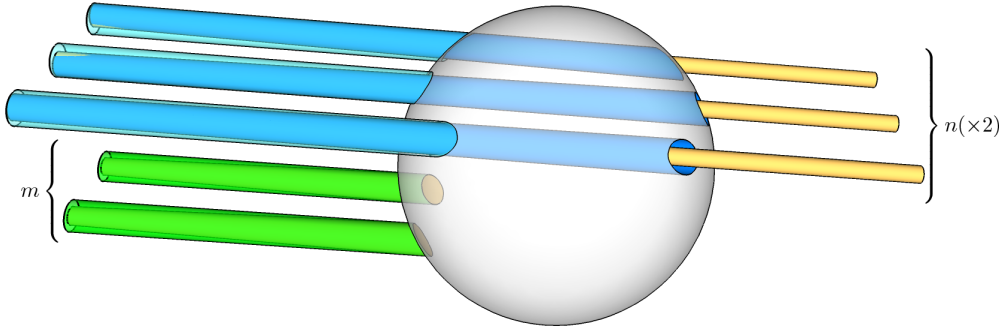
	\singleimage{Example of a sentence gadget with $m=2$ relations $(U,s)$ and $n=3$ relations $(U,u)$ with additional tubes.}{.8}{gadget_part}
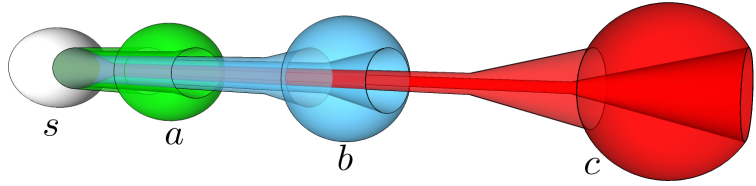
	\singleimage{Example of a $(U,s)$-gadget with $U=\{a,b,c\}$, with sentence gadgets $\{a,b,c,s\}$.}{.6}{gadget}

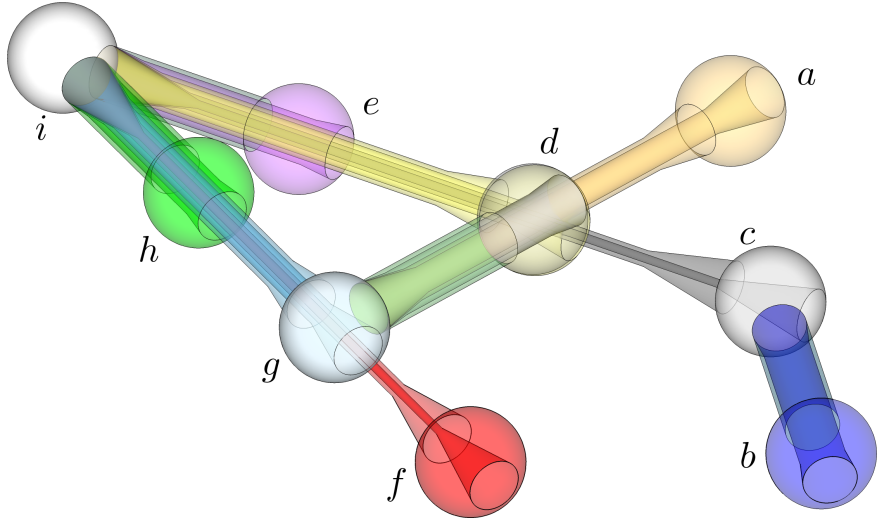
\singleimage{Constructing an instance of \textsc{Erasability} from an instance $(S,R)$ of \textsc{Minimum Axiom Set}, where $S=\{a,b,c,d,e,f,g,h,i\}$ and $R=\{(\{c,d,e\},i),(\{f,g,h\},i), (\{b\},c),(\{a,d\},g)\}$.}{.7}{gadget_construction}

Then, by construction the following holds for the $(U,s)$-gadget.
\begin{lemma} \label{lem:implication_gadget} 
	A $(U,s)$-gadget can be erased if and only if all sentence gadgets corresponding to sentences in $U$ are already erased.
\end{lemma}

\begin{proof}
	Clearly, if all sentence gadgets corresponding to sentences in $U$ are erased, the whole gadget can be erased tube by tube. If, on the other hand, one of the sentence gadgets still exists, this gadget together with the two tubes connected to it build a complex without external triangles which thus cannot be erased. 
\end{proof}

With these tools in mind we can now prove the main theorem of this section.
\begin{theorem}\label{thm:WPhard}
	\textsc{Minimum Axiom Set} $\leq_{FPT}$ \textsc{Erasability}, even when the instance of \textsc{Erasability} is a strongly connected pure 2-dimensional simplicial complex $\Delta$ which is embeddable in $\mathbb{R}^3$.  Therefore \textsc{Erasability} is $W[P]$-hard.
\end{theorem}

The simplicial complex $\Delta$ (\figref{gadget_construction}) constructed to prove $W[P]$-hardness of \textsc{Erasability} is in fact embeddable into $\mathbb{R}^3$. This means that, even in the relatively well behaved class of embeddable $2$-dimensional simplicial complexes, \textsc{Erasability} when bounding the number of critical simplexes is still likely to be inherently difficult. Please refer to the full version of this paper for a detailed proof of Theorem~\ref{thm:WPhard}.
%
The $W[P]$-completeness result implies that if \textsc{Erasability} turns out to be fixed parameter tractable, then $W[P]=FPT$, i.e., every problem in $W[P]$ including the ones lower in the hierarchy would turn out to be fixed parameter tractable, an unlikely and unexpected collapse in parameterized complexity. Also, it would imply that the $n$-variable SAT problem can be solved in time $2^{o(n)}$, that is, better than in a brute force search \cite{abrahamson1995fixed}. With respect to this result, if we want to prove fixed parameter tractability of \textsc{Erasability}, the parameter must be different from the natural parameter.


\section{Fixed parameter tractability in the treewidth} \label{sec:treewidth}

In this section, we prove that there is still hope to find an efficient algorithm to solve \textsc{Morse Matching}. We give positive results for the field of discrete Morse theory by proving that \textsc{Erasability} and \textsc{Morse Matching} are fixed parameter tractable in the treewidth of the spine of the input simplicial complex, and also in the dual graph of the problem instance in case it is a simplicial triangulation of a $3$-manifold.


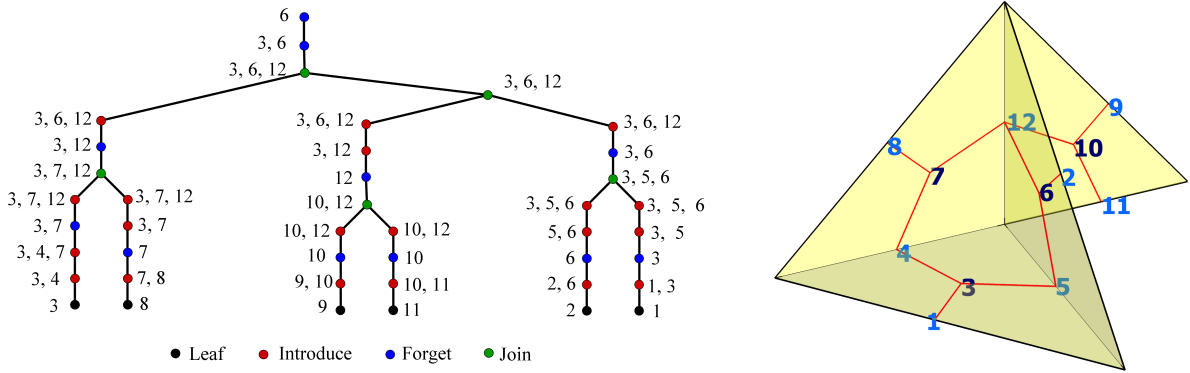
\tsubimages{Example of a nice tree decomposition (left) of the spine of a non-manifold $2$-dimensional simplicial complex (right).}{decomposition}{
  \vsubimage{.3}{big_nice_tree_decomposition}
  \vsubimage{.3}{decomposition3d}
}

\subsection{Treewidth}
\label{sec:tw}

\begin{definition}[Treewidth]
\label{def:treewidth} 
A tree decomposition of a graph $G$ is a tree $T$ together with a collection of bags $\{X_i\}$, where $i$ is a node of $T$. Each bag $X_i$ is a subset of nodes of $G$, and we require that (i)
every node of $G$ is contained in at least one bag $X_i$ ({\em node coverage}); (ii) for each arc of $G$, some bag $X_i$ contains both its endpoints ({\em arc coverage}); and for all bags $X_i$, $X_j$ and $X_k$ of $T$, if $X_j$ lies on the unique simple path from $X_i$ to $X_k$ in $T$, then $X_i \cap X_k \subseteq X_j$ ({\em coherence}). 

The \textit{width} of a tree decomposition is defined as $\max |X_i|-1$, and the \textit{treewidth} of $G$ is the minimum width over all tree decompositions. We will denote the \textit{treewidth} of $G$ by $\mathbf{tw}(G)$.
\end{definition}

For bounded treewidth, computing a tree decomposition of a graph
$G=(V,E)$ of width $\leq k$ has running time $O(f(k) |V|)$~\cite{bodlaender1993treedcomposition} due to an algorithm by
Bodlaender. Regarding the size of $f(k)$: using the improved algorithm
by Perkovi\'c and Reed~\cite{Perkovic99TreeDecAlgo}, at most $O(k^2)$
recursive calls of Bodlaender's improved linear time fixed-parameter
tractable algorithm for bounded treewidth from
\cite{Bodlaender96EfficientTreeDecAlgo} are needed.
This latter algorithm in turn is said
to have a constant factor $f(k)$ which is ``at most singly exponential
in $k$''. For further reading on the running times of tree decomposition
algorithms see
\cite{Bodlaender05DiscoveringTreewidth,kloks1994treewidth}.

\begin{definition}[Nice tree decomposition]
\label{def:nice_treewidth} 
A tree decomposition $({X_i \, | \, i \in I}, T)$ is called a {\em nice tree decomposition} if the following conditions are satisfied:
	\begin{enumerate}
	\item There is a fixed bag $X_r$ with $|X_r|=1$ acting as the root of $T$ (in this case $X_r$ is called the {\em root bag}).
	\item If bag $X_j$ has no children, then $|X_j|=1$ (in this case $X_j$ is called a {\em leaf bag}).
	\item Every bag of the tree $T$ has at most two children.
	\item If a bag $X_i$ has two children $X_j$ and $X_k$, then $X_i = X_j = X_k$ (in this case $X_i$ is called a {\em join bag}).
	\item If a bag $i$ has one child $j$, then either
		\begin{enumerate}
			\item $|X_i| = |X_j| + 1$ and $X_j \subset X_i$ (in this case $X_i$ is called an {\em introduce bag}), or
			\item $|X_j| = |X_i| + 1$ and $X_i \subset X_j$ (in this case $X_i$ is called a {\em forget bag}).
		\end{enumerate}
	\end{enumerate}
\end{definition}
	
A given tree decomposition can be transformed into a nice tree decomposition (\figref{decomposition}) in linear time:

\begin{lemma}[\cite{kloks1994treewidth}]
	Given a tree decomposition of a graph $G$ of width $w$ and $O(n)$ bags, where $n$ is the number of nodes of $G$, we can find a nice tree decomposition of $G$ that also has width $w$ and $O(n)$ bags in time $O(n)$.
\end{lemma}


\subsection{Alternating cycle-free matchings}
\label{ssec:acfm}

Given a graph $G = (N,A)$ and a matching $M \subset A$ on $G$, an \emph{alternating path} is a sequence of pairwise adjacent arcs such that each matched arc in the sequence is followed by an unmatched arc and conversely. An \emph{alternating cycle} of $M$ is a closed alternating path. Matchings which do not have any such alternating cycle are called \emph{alternating cycle-free matchings}. From the definition of Morse matching, we can state \textsc{Erasability} in the language of alternating cycle-free matchings as follows:

\begin{problem}[Alternating cycle-free matching]
{A bipartite graph $G=(N_1 \cup N_2,A)$.}
{A nonnegative integer $k$.}
{Does $G$ has an alternating cycle-free matching $M$ with at most $k$ unmatched nodes in $N_1$?}
\end{problem}

Specifically, if $G=H_1$ is the spine for some simplicial complex $\Delta$, then \textsc{Erasability} is equivalent to the \textsc{Alternating cycle-free matching} problem.


\subsection{FPT algorithm for the {alternating cycle-free matching} problem}

Courcelle's theorem~\cite{courcelle1990monadic} can be used to show that \textsc{Alternating cycle-free matching} is fixed parameter tractable (please refer to the full version of this paper). However, this is a purely theoretical result, since the stated complexity contains towers of exponents in the parameter function. This is the reason why, for the remainder of this section, we focus on the construction of a linear time algorithm to solve \textsc{Alternating cycle-free matching} for inputs of bounded treewidth with a significantly faster running time.

\begin{theorem}
	\label{thm:treewidth}
	Let $G = (N_1 \cup N_2, A)$ be a simple bipartite graph with a given nice tree decomposition $({X_i \, | \, i \in I}, T)$. Then the size of a maximum alternating cycle-free matching of $G$ can be computed in
$O( 4^{w^2+w} \cdot w^3 \cdot \log(w) \cdot n)$ time, where $n=|N|$ and $w$ denotes the width of the tree decomposition.
\end{theorem}

\paragraph*{Algorithm overview}
Our algorithm constructs alternating cycle-free matchings of $G$ along the nice tree decomposition $({X_i \, | \, i \in I}, T)$ of $G$, from the leaves up to the root, visiting each bag exactly once. In the following we will denote by $F_i$, the set of nodes which are already processed and forgotten by the time $X_i$ is reached; we call this the {\em set of forgotten nodes}. At each bag $X_i$ of the decomposition, we construct a set $\Mmatchset{i}$ representing {\em all} valid alternating cycle-free matchings in the graph induced by the nodes in $X_i \cup F_i$. 

The leaf bags contain a single node of $G$, and the only matching is thus empty. At each introduce bag $X_i=X_j\cup\{x\}$, each matching $\Mmatch$ of $\Mmatchset{j}$ can be extended to several matchings as follows. The newly introduced node $x$ can be either left unmatched, or matched with one of its neighbors as long as it generates a valid and cycle-free matching with $\Mmatch$. At each join bag $X_i=X_j=X_k$, $\Mmatchset{i}$ is build from the valid combinations of pairs of matchings from $\Mmatchset{j}$ and $\Mmatchset{k}$. The final list of valid matchings is then evaluated at the root bag $r$.

However, this final list $\Mmatchset{r}$ contains an exponential number of matchings. Fortunately, the nice tree decomposition allows us to group together, at each step, all matchings $\Mmatch$ that coincide on the nodes of $X_i$. Indeed, the algorithm takes the same decisions for all the matchings of the group. We can thus store and process a much smaller list $\Mmatchset{j}$ of matchings containing only one representative $\Mmatchr$ of each group. In each group, we choose one with the smallest number of unmatched nodes so far. This grouping takes place at the forget and join bags. This makes the algorithm exponential in the bag size, not the input size.
The algorithm is described step-by-step and illustrated in the full version of this paper.

\paragraph*{Matching data structure}
The structure storing an alternating cycle-free matching $\Mmatch$ in a set $\Mmatchset{i}$ must be suitable for checking the matching validity whenever a matching is extended at an introduce bag or a join bag. It must store which nodes are already matched in $\Mmatch$ to avoid matching a node of $G$ twice ({\em matching condition}). We use a binary vector $\mathbf{v}(\Mmatch)$, where the $x$-th coordinate is $1$ if node $x \in X_i$ is matched and $0$ otherwise. Checking the matching condition and updating when nodes are matched has thus a constant execution time $O(1)$.

Also, the structure must store which nodes are connected by an alternating path in $\Mmatch$ to avoid closing a cycle when extending or combining $\Mmatch$ ({\em cycle-free condition}). 
When matching two nodes $x$ and $y$, an alternating cycle is created if there exists an alternating path from a neighbor of $x$ to a neighbor of $y$. 
To test this, we use a \funname{union-find} structure~\cite{Tarjan1975unionfind} $\mathbf{uf}(\Mmatch)$, storing for each matched node $x$ the index of a matched node $c(x)$ connected to $x$ by an alternating path in $\Mmatch$. For a subset of matched nodes which are all connected to each other, the component index $c$ is chosen to be the node with the lowest index.
For each unmatched node, we store the ordered list of component indexes of neighbor matched nodes.
The cycle-free condition check reduces to \funname{find} calls on the
adjacent lists, and the update of the structure when increasing the
matching size reduces to \funname{union} calls, both executing in
near-constant time. %
All the matchings are stored in a hash structure to allow faster search for duplicates.
Finally, we can return not only the maximal cycle-free matching size, but an actual maximal cycle-free matching by storing, along with each representative matching, a binary vector of size $|X_i\cup F_i|$ with all the matched nodes so far.
 
\paragraph*{Grouping}
Traversing the nice tree decomposition in a bot\-tom-up fashion, each node appears in a set of bags that form a subtree of the tree decomposition (\emph{coherence requirement}). This means that, whenever a node is forgotten, it is never introduced again in the bottom-up traversal.

A na\"ive version of the algorithm described above would build the complete list of valid alternating cycle-free matchings: the set $\Mmatchset{i}$ would contain all valid matchings in the graph induced by the nodes in $X_i \cup F_i$. In particular, for each matching $\Mmatch \in \Mmatchset{i}$ the algorithm would store the binary vector $\mathbf{v}(\Mmatch)$ and the \funname{union-find} structure $\mathbf{uf}(\Mmatch)$ on $X_i \cup F_i$. However, it is sufficient to store the essential information about each $\Mmatch$ by restricting the \funname{union-find} structure $\mathbf{uf}(\Mmatch)$ and the binary vector $\mathbf{v}(\Mmatch)$ {\em only} to the nodes in the bag $X_i$ (for any matched node $x \in X_i$, node $c(x)$ of the \funname{union-find} structure is then chosen inside $X_i$). More precisely, we define an equivalence relation $\sim_i$ on the matchings of $\Mmatchset{i}$ such that $\Mmatch\sim_i\Mmatch'$ if and only if $\mathbf{v}(\Mmatch)=\mathbf{v}(\Mmatch')$ and $\mathbf{uf}(\Mmatch)=\mathbf{uf}(\Mmatch')$  on the nodes of $X_i$. Since two equivalent matchings only differ on the forgotten nodes $F_i$, the validation of the matching and cycle-free conditions of any extension of $\Mmatch$ or $\Mmatch'$  (or any combination with a third equivalent matching $\Mmatch''$) will be equal from now on.

Since we are interested in the alternating cycle-free matching with the minimum number of unmatched nodes, for each equivalence class we will choose a matching $\Mmatchr$ with the minimum number $m(\Mmatchr)$ of unmatched forgotten nodes as class representative. This number $m(\Mmatchr)$ is stored together with $(\mathbf{v},\mathbf{uf})$
for each equivalence class of $\Mmatchsetr{i}=\Mmatchset{i}/\!\!\sim_i$. In addition, we can compute the alternating cycle-free matching of maximum size by storing the complete binary vector $\mathbf{v}$ along with $m(\Mmatchr)$ (since the matching is cycle-free, this is sufficient to recover the set of arcs defining the matching).



\paragraph*{Execution time complexity}
To measure the running time we need to bound the number of equivalence classes of $\Mmatchsetr{i}$. Let $w_i$ be the number of nodes in $X_i$. The number of equivalence classes of $\Mmatchsetr{i}$ is then bounded above by the number of possible pairs $(\mathbf{v},\mathbf{uf})$ on $w_i$ nodes. The \funname{union-find} stores for each node $x$, either a component node $c(x)\in X_i$ or a list of at most $w_i$ component nodes, leading to at worst $2^{w_i}$ different lists per node, giving $2^{w_i^2}$ possible combinations of lists. Also there are $2^{w_i}$ possible binary vectors $\mathbf{v}$ of length $w_i$, therefore there are at worst $2^{w_i^2}2^{w_i}$ elements in $\Mmatchsetr{i}$ (this enumeration includes invalid matchings and incoherences between $\mathbf{v}$ and $\mathbf{uf}$). 

The time complexity is dominated by the execution at the join bag where pairs of equivalences classes from $\Mmatchsetr{j}$ and $\Mmatchsetr{k}$ have to be combined. Therefore we must square the number of equivalence classes in each set: the complexity for a join bag is  $O( 4^{w^2+w} \cdot w^3 \cdot \log(w))$ (please refer to the full version of this paper for details). Since there are $O(n)$ bags in a nice tree decomposition, the total execution time is in $O( 4^{w^2+w} \cdot w^3 \cdot \log(w)\cdot n)$. Finally, as already stated in Section~\ref{sec:tw}, for bounded treewidth computing a tree decomposition and a nice tree decomposition is linear.
So the whole process from the bipartite graph to the resulting maximal
alternating cycle-free matching is fixed-parameter tractable in the treewidth. Note that neither the decomposition nor the algorithm use the fact that the graph is bipartite.

\begin{table*}
\centering
\begin{tabular}{|c|rr||c|c|rr||c|c|rr|}
\hline
$n$ & \multicolumn{2}{l||}{$\#$ triangulations} & $\min w$ & $\max w$ & \multicolumn{2}{l||}{$\bar{w}$} & $\min w$ & $\max w$ & \multicolumn{2}{l|}{$\bar{w}$ (dual)} \\
\hline
$1$&$     4$&$(3)$&$    1$&$    2$&$ 1.50$&$(1.67)$&			$ 0$&$ 0$&$ 0.00$&\\
$2$&$    17$&$(12)$&$    1$&$    3$&$ 2.47$&$(2.42)$&			$ 1$&$ 1$&$ 1.00$&\\
$3$&$    81$&$(63)$&$    1$&$    3$&$ 2.51$&$(2.49)$&			$ 1$&$ 2$&$ 1.60$&$(1.52)$\\
$4$&$   577$&$(433)$&$    1$&$    5(4)$&$ 2.77$&$(2.73)$&		$ 1$&$ 3$&$ 1.91$&$(1.87)$\\
$5$&$  5184$&$(3961)$&$    1$&$    6(5)$&$ 2.95$&$ $&			$ 1$&$ 4$&$ 2.16$&$(2.18)$\\
$6$&$ 57753$&$(43584)$&$    1$&$    6$&$ 3.16$&$(3.19)$&		$ 1$&$ 4$&$ 2.31$&$(2.35)$\\
$7$&$ 722765$&$(538409)$&$    1$&$    7$&$ 3.35$&$(3.40)$&		$ 1$&$ 4$&$ 2.45$&$(2.50)$\\
\hline
\end{tabular}
\caption{Treewidths of the spine (left) and of the dual graphs (right) of closed generalized triangulations up to $7$ tetrahedra. The values in brackets are for $1$-vertex triangulations.}
\label{tab:closed}
\end{table*}

\subsection{Correctness of the Algorithm}
We must check that the algorithm, without the grouping, considers every possible alternating cycle-free matching in $G$ and that the grouping occurring at the forget and join bags does not discard the maximal matching.

The \textit{node coverage} and  \textit{arc coverage} properties of nice tree decompositions (Definition~\ref{def:treewidth}) ensure that each node is processed and each arc is considered for inclusion in the matching at one introduce node. Since the introduce node discards only matchings that violate either the matching or the cycle condition, and these violations cannot be legalized by further extensions or combinations of the matchings, all possible valid matchings are considered.

Now, consider two matchings $\Mmatch$ and $\Mmatch'$ that are grouped together and represented by $\Mmatchr$ at a forget or join bag $X_i$. In the further course of the algorithm, the representative $\Mmatchr$ is then extended or combined with other matchings to form new valid matchings $\Mmatchr'$. The \textit{coherence} property of Definition \ref{def:treewidth} assures that no neighbor of a newly introduced node can be a forgotten node, so the extension or combination only modifies matchings $\Mmatch$ and $\Mmatchr$ on nodes of $X_i$, which are represented in the structure of $\Mmatchr'$. Hence, the valid matchings $\Mmatchr'$ actually represent all the valid extensions and combinations of $\Mmatch$ and $\Mmatchr$. The grouping thus generates all valid and relevant representatives of matchings in order to find a maximal alternating cycle-free matching. Moreover, in case $\Mmatch$ and $\Mmatch'$ are equivalent and both with the lowest number of forgotten unmatched nodes, choosing $\Mmatch$ or $\Mmatch'$ as representative leads to the exact same extensions and combinations.

Finally, let $\Mmatch_m$ be the alternating cycle-free matching of maximum size of $G$. In each bag the corresponding matching must be one of the matchings with the lowest number of unmatched nodes within its equivalence class $\Mmatchr_m \in \Mmatchsetr{i}$. Otherwise, a matching in the same class $\Mmatchr_m$, extended and combined as $\Mmatch_m$ in the sequel of the algorithm would give rise to a matching with fewer unmatched nodes. Therefore, the choice of the representative at the forget and join bags never discards the future alternating cycle-free matching of maximum size.


\subsection{Treewidth of the dual graph}

Up to this point, we have been dealing primarily with simplicial complexes and their spines. We now turn our attention to simplicial triangulations of $3$-manifolds and a more natural parameter associated to them.

\begin{definition}[Dual graph] The \textit{dual graph} of a simplicial triangulation of a $3$-manifold $\mathcal{T}$, denoted $\Gamma(\mathcal{T})$, is the graph whose nodes represent tetrahedra of $\mathcal{T}$, and whose arcs represent pairs of tetrahedron faces that are joined together.
\end{definition}

We show that, if the treewidth of the dual graph is bounded,
so is the treewidth of the spine, as stated by the following theorem (please refer to the full version of this paper for the proof).

\begin{theorem} \label{thm:fpg}
	Let $G$ be the spine of a simplicial triangulation of a $3$-manifold $\mathcal{T}$. If $\mathbf{tw}(\Gamma(\mathcal{T})) \leq k$, then $\mathbf{tw}(G) \leq 10k+9$.
\end{theorem}

\begin{table*}
\centering
\begin{tabular}{|c|r||rc|rc|rr||rc|rc|rr|}
\hline
$n$ & $\#$ triangulations & \multicolumn{2}{l|}{$\min w$} & \multicolumn{2}{l|}{$\max w$} & \multicolumn{2}{l||}{$\bar{w}$} & \multicolumn{2}{l|}{$\min w$} & \multicolumn{2}{l|}{$\max w$} & \multicolumn{2}{l|}{$\bar{w}$ (dual)} \\
\hline
$5$&$     1$&&$    6$&&$    6$&&$ 6.00$&&				$ 4$&&$ 4$&&$ 4.00$\\
$6$&$     2$&$    \leq$&$ 7$&$    \leq$&$ 8$&$ \leq$&$ 7.50$&&				$ 4$&&$ 5$&&$ 4.50$\\
$7$&$     5$&$    \leq$&$ 8$&$    \leq$&$ 11$&$ \leq$&$ 9.40$&&				$ 4$&&$ 6$&&$ 5.00$\\
$8$&$     39$&$    \leq$&$ 8$&$    \leq$&$ 14$&$ \leq$&$ 11.23$&&				$ 4$&&$ 7$&&$ 5.74$\\
$9$&$     1297$&$    \leq$&$ 8$&$    \leq$&$ 18$&$ \leq$&$ 13.55$&&			$ 4$&&$ 9$&&$ 7.01$\\
$10$&$     249015$&$    \leq$&$ 8$&$    \leq$&$ 22$&$ \leq$&$ 16.33$&			$ \leq$&$ 4$&$ \leq$&$ 13$&$ \leq$&$ 8.99$\\
\hline
\end{tabular}
\caption{Upper bounds and exact values for the treewidths of the spine (left) and
of the dual graph (right) of simplicial triangulations of $3$-manifolds up to $10$ vertices.}
\label{tab:comb}
\end{table*}

\section{Experimental Results}\label{sec:experimental}

In Section \ref{sec:WpCompleteness} we have seen that the problem of finding optimal Morse matchings is hard to solve in general. In Section \ref{sec:treewidth} on the other hand we proved that in the case of a small treewidth of the spine of a $2$-dimensional complex or, equivalently, in the case of a bounded treewidth of the dual graph of a simplicial triangulation of a $3$-manifold, finding an optimal Morse matching becomes easier. Up to a certain scaling factor, the results stated in Section \ref{sec:treewidth} hold for generalized triangulations as well (also, note that the notion of a spine or the dual graph can be extended in a straightforward way to generalized triangulations).

Given this situation, a natural question to ask is the following: What is a {\em typical} value for the treewidth of the respective graphs of (i) small generic generalized triangulations of $3$-manifolds, and (ii) small generic simplicial triangulations of $3$-manifolds?

In a series of computer experiments we computed the treewidth of the relevant graphs (i.e., the spine and the dual graph) of all closed generalized triangulations of $3$-manifolds up to $7$ tetrahedra \cite{burton11-genus}, and all simplicial triangulations of $3$-manifolds up to $10$ vertices \cite{Lutz08ThreeMfldsWith10Vertices}.
The computer experiments were done using \funname{LibTW}
\cite{Dijk06TreewidthDotCom} to compute the treewidth / upper bounds for
the treewidth, with the help of the \funname{GAP} package \textsf{simpcomp} \cite{simpcompISSAC,simpcomp} and the $3$-manifold software \funname{Regina} \cite{burton04-regina,regina}. We report the minimal, maximal and average treewidths of all triangulations with the same number of vertices in Tables \ref{tab:closed} and \ref{tab:comb}.

\medskip
Regarding the treewidth of generalized triangulations of $3$-manifolds, we observe that there is a large difference between the average treewidth and the maximal treewidth for both the dual graph and the spine. In particular, the average treewidth appears to be relatively small. Moreover, there is only a slight difference between the data for general closed triangulations and $1$-vertex triangulations. This fact is somehow in accordance with our intuition since the number of $0$-dimensional simplexes should neither directly affect the spine nor the dual graph of a generalized triangulation.

On the other hand, the gap between the maximum treewidth and the average treewidth in the case of simplicial triangulations of $3$-manifolds is relatively small compared to the data for generalized triangulations. In addition, the treewidth of the spines of some particularly interesting $2$-dimensional simplicial complexes (reported in the full version of this paper) is significantly smaller than the (upper bound of the) treewidth of simplicial triangulations of $3$-manifolds. At this point it is important to note that, while the data concerning the spines for simplicial complexes only consists of upper bounds, experiments applying the algorithm for the upper bound to smaller graphs and then computing their real treewidths suggest that these upper bounds (in average) are reasonably close to the exact treewidth. 

Further analysis shows that the average treewidth of the spines for both generalized and simplicial triangulations of $3$-manifolds is mostly less than twice the treewidth of the dual graph, and hence much below the theoretical upper bound given by Theorem \ref{thm:fpg}. Also, the ratio between these two numbers appears to be more or less stable for all values shown in Tables \ref{tab:closed} and \ref{tab:comb}. This can be seen as experimental evidence that for triangulated $3$-manifolds the treewidth of the dual graph is responsible for the inherent difficulty to solve \textsc{Erasability} and related problems.

\medskip
Despite the small values of $n$ in our tables,
there are theoretical reasons to believe that the patterns of small treewidth will continue for larger $n$.
For instance, the conjectured minimal triangulations of Seifert fibered spaces over the sphere have dual graphs with $O(1)$ treewidth for arbitrary $n$.
Moreover, following recent results of Gabai~\textit{et al.} \cite{Gabai09MinVolCuspedHyperb3Mflds} there are reasons to believe that large infinite classes of topological $3$-manifolds
admit triangulations whose treewidths are below provable upper bounds.  Investigating these upper bounds is work in progress.




\section{Acknowledgments}
This work is partially financed by CNPq, FAPERJ, PUC-Rio, CAPES, and Australian Research Council Discovery Projects DP1094516 and DP110101104.
We would also like to thank Michael Joswig for fruitful discussions.



\appendix
%
\newpage
\noindent
{\Huge \bf Appendix}

\section{Proof of Lemma~\ref{lem:morse_manifold}}\label{app:morse_manifold} 

The proof of Lemma~\ref{lem:morse_manifold} actually follows directly from Joswig and Pfetsch's proof of the following lemma.

\begin{lemma}[\cite{joswig2006computing}] 
	Let $M$ be a Morse matching on a simplicial complex $\Delta$. Then we can compute a Morse matching $M'$ in polynomial time which has exactly one critical $0$-simplex, the same number of critical simplexes of dimension greater or equal $2$ as $M$, and $c(M') \leq c(M)$.
\end{lemma}

The proof builds a Morse matching from a spanning tree of the primal graph, i.e. the graph obtained considering only the vertices and edges of $\Delta$. For a $3$-manifold $\Delta$, the proof of the previous lemma can be applied exactly the same way on the dual graph of $\Delta$ to obtain the following result.
\begin{lemma} 
Let $M$ be a Morse matching on a closed triangulated $3$-manifold $\Delta$. Then we can compute a Morse matching $M'$ in polynomial time which has exactly one critical $3$-simplex, the same number of critical simplexes of dimension less or equal $1$, and $c(M') \leq c(M)$.
\end{lemma}

Since the proof works independently on the primal and dual graph, Lemma \ref{lem:morse_manifold} is a combination of these results. Here, we simply reproduce the proof of Joswig and Pfetsch~\cite{joswig2006computing} verbatim applying it to $3$-manifold complexes, using Poincar\'{e}'s duality. 

First consider a Morse matching $M$ for a connected $3$-manifold $\Delta$. Let $\gamma(M)$ be the graph obtained from the primal graph of $\Delta$ by removing all arcs (edges of $\Delta$) matched with triangles in $M$ and let $\gamma^\ast(M)$ be obtained from the dual graph of $\Delta$ by removing all the arcs (triangles in $\Delta$) where the corresponding triangles are matched with edges of $\Delta$ in $M$. Note that $\gamma(M)$ and $\gamma^\ast(M)$ contain all vertices and tetrahedra of $\Delta$, respectively. 

\begin{lemma} 
The graph $\gamma(M)$ and dual graph $\gamma^\ast(M)$ are connected. 
\end{lemma}

\begin{proof}
Suppose that $\gamma(M)$ is disconnected. Let $N$ be the set of nodes in a connected component of $\gamma(M)$, and let $C$ be the set of cut edges, that is, edges of $\Delta$ with one vertex in $N$ and one vertex in its complement. Since $\Delta$ is connected, $C$ is not empty. By definition of $\gamma(M)$, each edge in $C$ is matched to a unique $2$-simplex.

Consider the directed subgraph $D$ of the Hasse diagram consisting of the edges in $C$ and their matching $2$-simplexes. The standard direction of arcs in the Hasse diagram (from the higher to the lower dimensional simplexes) is reversed for each matching pair of $M$, i.e., $D$ is a subgraph of $H(M)$. We construct a directed path in $D$ as follows. Start with any node of $D$ corresponding to a cut edge $e_1$. Go to the node of $D$ determined by the unique $2$-simplex $\tau_1$ to which $e_1$ is matched to. Then $\tau_1$ contains at least one other cut edge $e_2$, otherwise $e_1$ cannot be a cut edge. Now iteratively go to $e_2$, then to its unique matching 2-simplex $\tau_2$, choose another cut edge $e_3$, and so on. We observe that we obtain a directed path $e_1,\tau_1,e_2,\tau_2, \cdots$ in $D$, i.e., the arcs are directed in the correct direction. Since we have a finite graph at some point the path must arrive at a node of $D$ which we have visited already. Hence, $D$ (and therefore also $H(M)$) contains a directed cycle, which is a contradiction since $M$ is a Morse matching.

To prove that $\gamma^\ast(M)$ is connected, we repeat the proof above on the dual graph.
\end{proof}


\setcounter{lemma}{0}
\begin{lemma}[\cite{lewiner2004apps, joswig2006computing, ayala2012perfect}] 
	Let $M$ be a Morse matching on a triangulated $3$-manifold $\Delta$. Then we can compute a Morse matching $M'$ in polynomial time which has exactly one critical $0$-simplex and one critical $3$-simplex, such that $c(M') \leq c(M)$.
\end{lemma}

\begin{proof}
Since $\gamma(M)$ and $\gamma^\ast(M)$ are connected, they both have spanning trees, and we will use them to build the Morse matching. First pick an arbitrary node $r_1$ and any spanning tree of $\gamma(M)$ and direct all arcs away from $r_1$. Then pick an arbitrary tetrahedron (a node in the dual graph) $r_2$ and any spanning tree of $\gamma^\ast (M)$ and direct all triangles (arcs in dual graph) away from $r_2$. This yields a maximum Morse matching on $\gamma(M)$ and $\gamma^\ast (M)$. It is easy to see that replacing the part of $M$ on $\gamma(M)$ and $\gamma^\ast (M)$ with this matching yields a Morse matching. This Morse matching has only one critical vertex (the root $r_1$) and one critical tetrahedron (the root $r_2$). Note that every Morse matching in a triangulated $3$-manifold contains at least one critical vertex and at least one critical tetrahedron; this can be seen from Theorem \ref{forman_theorem}. Furthermore, the total number of critical simplexes can only decrease, since we computed an optimal Morse matching on $\gamma(M)$ and $\gamma^\ast (M)$.
\end{proof}


\newpage
\section{Proof of Theorem~\ref{thm:inWP}}\label{app:inWP}

\setcounter{theorem}{2}
\begin{theorem}
	\textsc{Erasability} $\leq_{FPT}$ \textsc{Minimum Axiom Set}, therefore \textsc{Erasability} is in $W[P]$.
\end{theorem}

\begin{proof}
We show membership of \textsc{Erasability} in $W[P]$ by reducing a given instance $(\Delta,k)$ of \textsc{Erasability} to an instance $(S,R,k)$ of \textsc{Minimum Axiom Set}. 

W.~l.~o.~g. we can assume that the $2$-dimensional simplicial complex $\Delta$ has no external edges (if $\Delta$ has external edges we first remove these edges until no external edge exists and reduce the remaining problem instance to an instance of \textsc{Minimum Axiom Set}). We now identify the set of triangles of $\Delta$ with the set of sentences $S$ in a one-to-one correspondence. For every edge $e \in \Delta$ we denote the set of all triangles containing $e$ by $\operatorname{star}_{\Delta} (e) \subset \Delta$, we write for the corresponding set of sentences $S_e \subset S$, and we define the set of implication relations $R$ by the relations 
$$ (S_e \setminus \{s\}, s)$$
for each triangle $s \in S_e$ for all edges $e \in \Delta$. Note that $\Delta$ has no external edges and thus $S_e \setminus \{s\} \neq \emptyset$ for all $e$.

\medskip
In a next step, we show that for all axiom sets $S_0 \subset S$ of size $k$ we have $ \Delta \setminus \Delta_0 \rightsquigarrow \emptyset$ for the associated subset of triangles $\Delta_0 \subset \Delta$ of size $k$. To see that this is true, note that for the augmenting sequence $S_0 \subset S_1 \subset \ldots \subset S_n = S$ of $S$, their corresponding subsets of triangles $\Delta_0 \subset \Delta_1 \subset \ldots \subset \Delta_n = \Delta$ and $i \in \{ 1 , \ldots,  n \}$ fixed, all sentences $s \in S_i \setminus S_{i-1}$ have to occur in a relation $(S_e \setminus \{s\}, s)$ for some edge $e$ with $S_e \setminus \{ s \} \subset S_{i-1}$. For the triangle $t \in \Delta$ corresponding to $s$ this means that, $\operatorname{star}_{\Delta} (e) \setminus t \subset \Delta_{i-1}$. Thus, if we assume that all triangles in $\Delta_{i-1}$ are already erased, $t$ must be external and thus can be erased as well. The statement now follows by the fact that for $i=1$, all triangles in $\Delta_0$ are already erased in $\Delta \setminus \Delta_0$ and hence $ \Delta \setminus \Delta_0 \rightsquigarrow \emptyset$.

\medskip
Conversely, let $\Delta_0 \subset \Delta$ be of size $k$ such that $ \Delta \setminus \Delta_0 \rightsquigarrow \emptyset$. Since $\Delta$ has no external triangles but $ \Delta \setminus \Delta_0 \rightsquigarrow \emptyset$, there must be external triangles $ t \in \Delta \setminus \Delta_0 $ and hence for $s\in S$ being the sentence corresponding to the triangle $t$ there is a relation $(S_e \setminus \{s\}, s)$ with $S_e \setminus \{s\} \subset S_0$, where $S_0$ is the set of sentences corresponding to the set of triangles $\Delta_0$. We then define $S_1$ to be the union of $S_0$ with all sentences $s$ of the type described above and iterating this step results in a sequence of subsets $S_0 \subset S_1 \subset \ldots \subset S_n = S$ for some $n$ what proves the result.
\end{proof}


\newpage
\section{Proof of Theorem~\ref{thm:WPhard}}\label{app:WPhard}

\setcounter{theorem}{3}
\begin{theorem}
	\textsc{Minimum Axiom Set} $\leq_{FPT}$ \textsc{Erasability}, even when the instance of \textsc{Erasability} is a strongly connected pure 2-dimensional simplicial complex $\Delta$ which is embeddable in $\mathbb{R}^3$.  Therefore \textsc{Erasability} is $W[P]$-hard.
\end{theorem}

\begin{proof}
To show $W[P]$-hardness of \textsc{Erasability}, we will reduce an arbitrary instance $(S,R,k)$ from \textsc{Minimum Axiom Set} to an instance $(\Delta,k)$ of \textsc{Erasability}. In order to do so, we will use a sentence gadget for each element of $S$ and an implication gadget for each relation $R$ (cf.\ Definition \ref{def:implication_gadget}) to construct a $2$-dimensional simplicial complex $\Delta$ with a polynomial number of triangles in the input size.

By construction, we can glue all sentence and implication gadgets together in order to obtain a simplicial complex $\Delta$ without any exterior triangles. Note that the only place where $\Delta$ is not a surface is at the former $m$ boundary components per sentence gadget corresponding to the $m$ relations in $R$ with the respective right hand side.

For any axiom set $S_0 \subset S$ of size $k$, let $\Delta_0$ be a set of $k$ triangles, one from each sentence gadget corresponding to a sentence in $S_0$. It follows by Lemma \ref{lem:implication_gadget}, that $\Delta \setminus \Delta_0$ can be erased to a complex where all the sentence gadgets $s$ corresponding to relations $(U,s)$, $U \subset S_0$, have external triangles. Since $S_0$ is an axiom set, iterating this process erases the whole complex $\Delta$.

\medskip
Conversely, let $\Delta_0$ be a set of $k$ triangles such that $\Delta \setminus \Delta_0 \rightsquigarrow \emptyset$. First, note that erasing a triangle of any tube of an implication gadget always allows us to remove the sentence gadget at the right end of this tube. Hence, w.~l.~o.~g.\ we can assume that all $k$ triangles in $\Delta_0$ are triangles of some sentence gadget in $\Delta$. Now, if any sentence gadget contains more than one triangle of $\Delta_0$ we delete all additional triangles obtaining a set $\Delta_0'$ of $k'$ triangles, $k' \leq k$, such that $\Delta \setminus \Delta_0' \rightsquigarrow \emptyset$ and thus the corresponding set of sentences is an axiom set of size $k' \leq k$.

\medskip
The result now follows by the observation that $\Delta$ can be realized by at most a quadratic number of triangles in the input size of $(S,R,k)$.
\end{proof}


\newpage
\section{Fixed parameter tractability of \textsc{Alternating cycle-free matching} from Courcelle's theorem}\label{app:courcelle}

Courcelle's celebrated theorem \cite{courcelle1990monadic} states that all graph properties that can be defined in {\em Monadic Second-Order Logic} (MSOL) can be decided in linear time when the graph has bounded treewidth. Here, we want to use Courcelle's theorem to show that problems in discrete Morse theory are fixed parameter tractable in the treewidth of some graph associated to the problem. However, it is not obvious how to {\em directly} state \textsc{Erasability} and \textsc{Morse Matching} in MSOL. Instead, we will apply Courcelle's theorem to \textsc{Alternating cycle-free matching} which by the comment made in Section~\ref{ssec:acfm} is a graph theoretical problem equivalent to \textsc{Erasability}.

\begin{theorem} 
	\label{thm:courcelle}
	Let $w \geq 1$. Given a bipartite graph with $\mathbf{tw}(G) \leq w$, \textsc{Alternating cycle-free matching} can be solved in linear time.
\end{theorem}

\begin{proof}
	We will write a MSOL formulation of \textsc{Alternating cycle-free matching} based on the fact that $M$ is an alternating cycle-free matching if and only if $M$ is a matching and every induced $M$-subgraph contains a node of degree $1$ \cite{golumbic2001uniquely}:
\begin{multline*} 
	\max M:
	\forall x \in N [\neg\exists a_1,a_2 \in M (a_1 \neq a_2 \land inc(x,a_1) \land inc(x,a_2))] \\
	\phantom{Mm} \land \forall M' \subseteq M (\exists a \in M', \exists x \in N[inc(x,a) \land (\forall x_1 \in N(  \neg \exists a_1 \in M'( x \neq x' \land adj(x,x_1) \land inc(x_1,a_1))))])
\end{multline*}
where $inc(x,a)$ is the incidence predicate between node $x$ and arc $a$ and $adj(x,x')$ is the adjacency predicate between node $x$ and node $x'$. The above statement can be translated to plain English as follows: ``Find the largest matching $M$ of $G$, where each node is incident to at most one arc, such that in every subset $M'$ of the matching $M$ there exists a matched node $x$ in $M'$ such that its only neighbor matched in $M'$ is the other endpoint of the unique matched arc incident to $x$.
\end{proof}


\newpage
\section{Algorithm for \textsc{Alternating cycle-free matching}: step by step}\label{app:algo}
The algorithm visits the bags of the nice tree decomposition bottom-up from the leaves to the root evaluating the corresponding mappings in each step according to the following rules (\figref{algorithm_details_rot}).

\paragraph*{Leaf bag}
The set of matchings $\Mmatchsetr{i}$ of a leaf bag $X_i=\{x\}$ is trivial with a unique empty matching $\Mmatchr$ represented by $\mathbf{v}(\Mmatchr) = [0]$, and $\textbf{uf}(\Mmatchr)(x)$ defined by as an empty list, associated with $m(\Mmatchr)=0$.

\paragraph*{Introduce bag}
Let $X_i=X_j\cup\{x\}$ be an introduce bag with child bag $X_j$. The set of valid matchings $\Mmatchsetr{i}$ is built from $\Mmatchsetr{j}$ by introducing $x$ in each matching $\Mmatchr \in \Mmatchsetr{j}$, generating several possible matchings $\Mmatchr'$.
We can always introduce $x$ as an unmatched node, then $\Mmatchr$ is extended on $x$ by setting $\mathbf{v}(\Mmatchr')_{|x} = 0$ and updating $\textbf{uf}(\Mmatchr')$ with the ordered list of components for each matched neighbor of $x$.
In addition, for each unmatched neighbor $y \in X_j$, we can introduce $x$ as a matched node in the following way. We match both $x$ and $y$ in $\Mmatchr$ and set $\mathbf{v}(\Mmatchr')_{|x} = 1$ and $\mathbf{v}(\Mmatchr')_{|y} = 1$. If the intersection of the list of neighbor components of $x$ and $y$ is empty, then the matching of $x$ and $y$ does not create a cycle. In this case $\Mmatchr'$ is a valid extension of $\Mmatchr$. The update of the \funname{union-find} structure must then reflect the extensions of all alternating paths through arc $\{x,y\}$. We perform in $\textbf{uf}(\Mmatchr')$ a $\funname{union}$ operation for $x$ and all its matched neighbors (including $y$), and for $y$ and all its matched neighbors. We also add the merged component index $c(x)$ to the list of neighbor components of each unmatched neighbor of $x$ and $y$.
Then we include all valid extensions $\Mmatchr'$ to $\Mmatchsetr{i}$, reducing $\mathbf{v}(\Mmatchr')$ by calling \funname{find} for each node and neighbor component list entry, and we set $m(\Mmatchr)=m(\Mmatchr')$ for all extensions $\Mmatchr'$ of $\Mmatchr$.

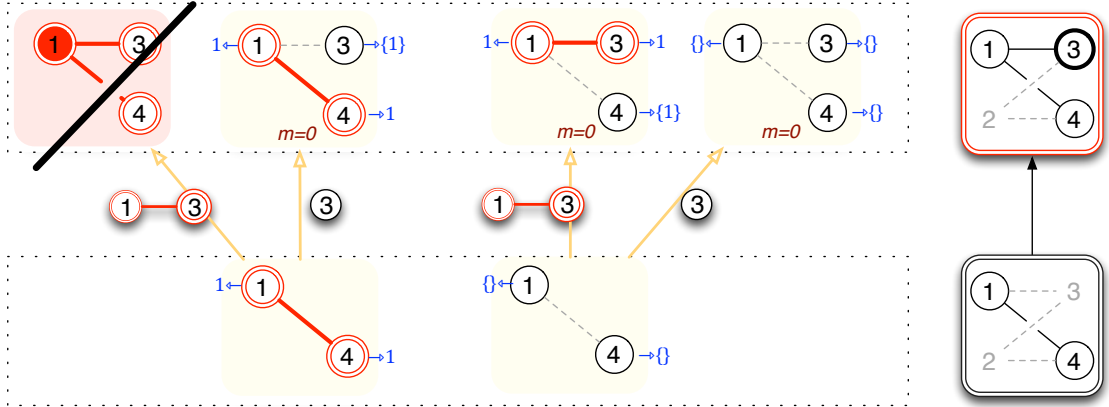
\singleimage{Details of the decisions at an introduce bag.}{.9}{algorithm_details_introduce}
\paragraph*{Running time}
There are at most $2^{w_i^2+w_i}$ extended matchings $\Mmatchr'$ for bag $X_i$ (including all invalid ones), where $w_i=|X_i|=|X_j|+1$ (a new possible matching can be generated only once). Each new matching is validated by a direct lookup at  $\mathbf{v}(\Mmatchr')$ and ordered list comparison, leading to a linear time $w_i$. The update of each structure requires constant time for each matched neighbor of $x$ and almost linear time $O(w_i)$ plus the sorted insertion $O(w_i\cdot \log(w_i))$ for each unmatched neighbor, and there are at most $w_i$ neighbors in the bag. Thus, the total running time of an introduce bag is in $O(2^{w_i^2+w_i} \cdot w_i^2 \cdot \log(w_i) )$.

\paragraph*{Forget bag}
Let $X_i=X_j\setminus\{x\}$ be a forget bag with child bag $X_j \ni x$. While the set of all possible matchings on $X_i \cup F_i$ does not change ($\Mmatchset{j}=\Mmatchset{i}$), the equivalence relation $\sim_i$ possibly identifies more matchings than $\sim_j$. For each matching $\Mmatchr \in \Mmatchsetr{j}$, a new matching $\Mmatch'$ is obtained by deleting coordinate $x$ of $\mathbf{v}(\Mmatchr)$. If $c(x)=x$, $\textbf{uf}(\Mmatchr)$ needs to be updated. To do so, the set of nodes $X_i$ is traversed twice, once to look for node $y\neq x$ of minimal index such that $c(y)=c(x)$ (eventually, $y$ is empty), and a second time to replace $x$ by $y$ each time $x$ is used as a component index. If $x$ was unmatched in $\Mmatchr$ (i.e., $\mathbf{v}(\Mmatchr)_{|x}=0$), then we set $m(\Mmatch')=m(\Mmatchr)+1$, otherwise we set $m(\Mmatch')=m(\Mmatchr)$.
Once the set $\Mmatchset{j}'$ of all the newly generated $\Mmatch'$ is computed, $\Mmatchset{i}$ is obtained as the quotient of $\Mmatchset{j}'$ by $\sim_i$, the equivalence relation on $X_i$. More precisely, each pair $(\Mmatch',\Mmatch'')\in \Mmatchset{j}'^2$ is tested for equality on both $\mathbf{v}$ and $\mathbf{uf}$. If they are equal, the one with the lowest $m$ is defined to be the new representative in $\Mmatchset{i}$.

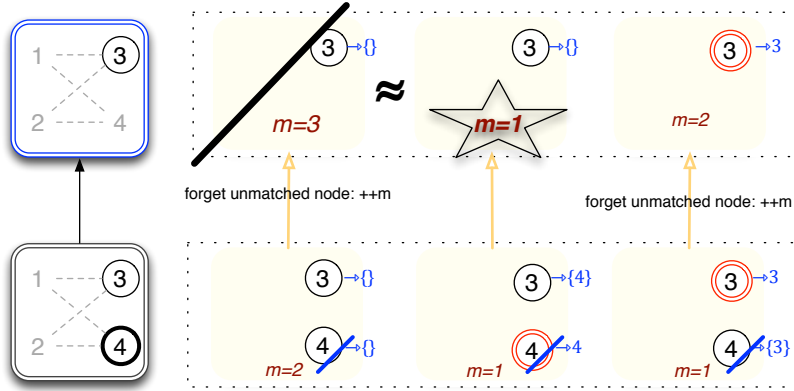
\singleimage{Details of the decisions at a forget bag.}{.65}{algorithm_details_forget}
\paragraph*{Running time}
Each new matching $\Mmatch'$ is obtained from a single element of $\Mmatchset{j}$ in worst-case time $O(w_i^2\cdot \log(w_i))$. Equivalent matchings are detected on-the-fly when filling the hash structure of $\Mmatchset{i}$, and each equivalence test is linear in $w_i^2$. The complexity is thus in $O( 2^{w_j^2+w_j} \cdot w_j^2 \cdot \log(w_j) )$.

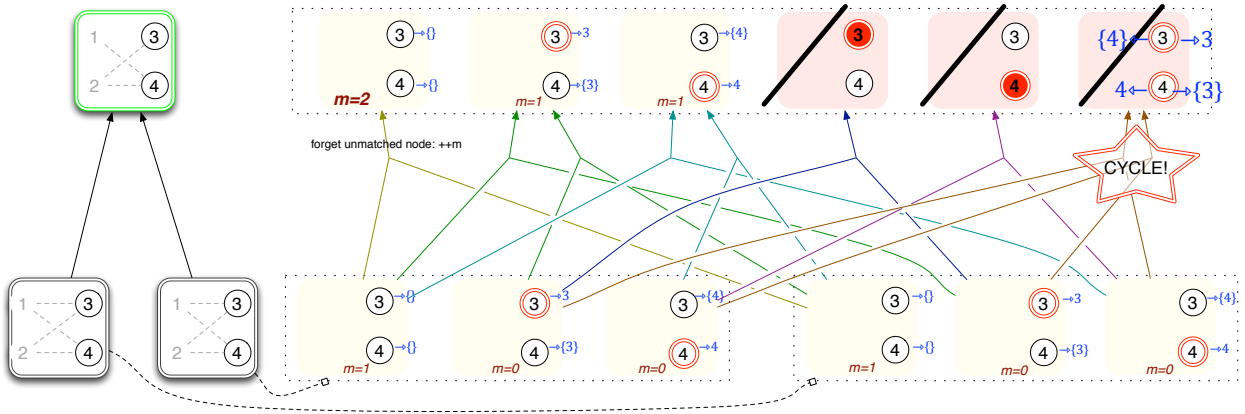
\singleimage{Details of the decisions at a join bag.}{1}{algorithm_details_join}
\paragraph*{Join bag}
Let $X_i=X_j=X_k$ be a join bag with child bags $X_j$ and $X_k$. The matchings of $\Mmatchset{i}$ are generated by combining all the pairs of matchings $(\Mmatch,\Mmatch')\in \Mmatchset{j} \times \Mmatchset{k}$. A combination is valid if and only if it satisfies both the matching and cycle-free conditions. The matching condition says that a node cannot be matched in both $\Mmatch$ and $\Mmatch'$, which is checked by a logical $AND$ operation ($\mathbf{v}(\Mmatch)\  AND\  \mathbf{v}(\Mmatch')$). The cycle-free condition is checked with the \funname{union-find} structures $\Mmatch$ and $\Mmatch'$: the combination is valid if no node of the component of a matched node in $\mathbf{uf}(\Mmatch)$ is a neighbor of the same component in $\mathbf{uf}(\Mmatch')$ and vice versa, each test requiring $O(w_i^2)$ per component.
If a combination is valid, its structure $\Mmatch''$ is defined by $\mathbf{v}(\Mmatch'') = \mathbf{v}(\Mmatch)\  OR\ \mathbf{v}(\Mmatch')$. The \funname{union-find} structure is initialized from $\mathbf{uf}(\Mmatch)$, and updated as the introduce bag for each matched nodes of $\Mmatch'$. Finally, $m(\Mmatch'') = m(\Mmatch)+m(\Mmatch')$.
As in the forget bag, two combinations may result in equivalent matchings, and we must compare them pairwise and choose the representative with the lowest number of unmatched forgotten bags.
Note that the sets of forgotten nodes of $X_j$ and forgotten nodes of $X_k$ have to be disjoint by the \textit{coherence} of Definition \ref{def:treewidth} and hence no forgotten node can be matched twice in this setting. Furthermore, all possible combinations of matched and unmatched nodes are enumerated in $\Mmatchset{j}$ and $\Mmatchset{k}$ and hence no possible matching is overseen.

\paragraph*{Running time}
Each pair of matchings is validated and updated in time $O(w_i\cdot w_i^2\cdot\log(w_i))$. The comparison and the choice of representative is done on-the-fly when filling the hash structure of $\Mmatchset{i}$. There are at worst $(2^{w_i^2+w_i})^2$ pairs. Thus, the complexity of the join bag dominates all other running times. Therefore, the complexity of the algorithm is in $O( 4^{w_i^2+w_i} \cdot w_i^3 \cdot \log(w_i) )$ per bag.

\paragraph*{Root bag}
Let $X_r=\{x\}$ be the root of $T$. $\Mmatchset{r}$ consists of at most two matchings $\mathbf{v}(\Mmatchr) = [0]$ or $\mathbf{v}(\Mmatchr') = [1]$, where $\textbf{uf}(\Mmatchr)$ is an empty list and $\textbf{uf}(\Mmatchr')$ is defined by $c(x)=x$. It follows that the minimum number of unmatched nodes for any alternating cycle-free matching of $G$ is given by  $m = \min \{ m(\Mmatchr) + 1 , m(\Mmatchr') \}$, and the maximum size of an alternating cycle-free matching is given by $(n - m)/2 $ where $n = |N|$ denotes the number of nodes of $G$.
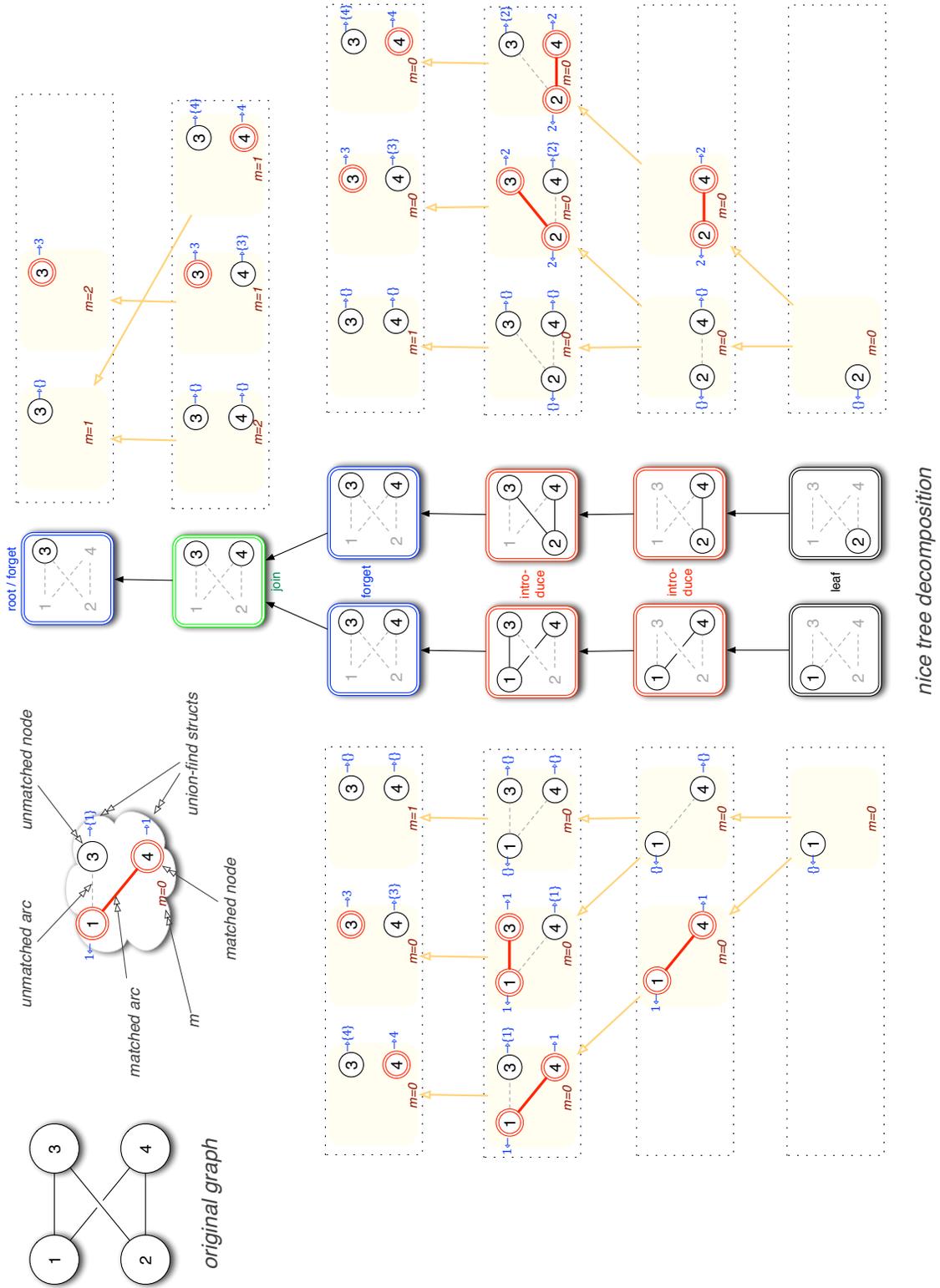
\singleimage{Algorithm execution on a small bipartite cycle (top left) with its nice tree decomposition (center). At each bag, a set of matchings $\Mmatchsetr{i}$ is generated according to the bag type. $\Mmatchsetr{i}$ is represented on the side of each bag, with the nomenclature illustrated at the top of the figure.}{.9}{algorithm_details_rot}

\paragraph*{Total Running Time}
The total time complexity of the algorithm per bag is bounded above by the running time of the join bag. Since there is a linear number of bags, and since for every bag $X_i$ we have $|X_i| \leq tw(G) + 1 = w + 1$, the total time complexity of the algorithm described above is 
$$O( 4^{w^2+w} \cdot w^3 \cdot \log(w) \cdot n) .$$

\clearpage
\section{Proof of Theorem~\ref{thm:fpg}}\label{app:fpg}

\setcounter{theorem}{5}
\begin{theorem}
	Let $G$ be the spine of a simplicial triangulation of a $3$-manifold $\mathcal{T}$. If $\mathbf{tw}(\Gamma(\mathcal{T})) \leq k$, then $\mathbf{tw}(G) \leq 10k+9$.
\end{theorem}

\begin{proof}
Let $T$ be a tree decomposition of the dual graph, where each bag $X_i$ contains less or equal $k+1$ tetrahedra. We show how to construct a tree decomposition $T'$ of the spine of $\mathcal{T}$, modeled on the same underlying tree as $T$, in which each bag $X'_i$ contains less or equal $10(k+1)$ edges and triangles.

For each bag $X_i$ of $T$, we simply define the bag $X'_i$ to contain all edges and triangles of all tetrahedra in $X_i$.  It remains to verify the three properties of a tree decomposition (Definition \ref{def:treewidth}).

\paragraph*{Node coverage}
It is clear that every edge or triangle in the spine belongs to some bag $X'_i$, since every edge or triangle is contained in some tetrahedron $\delta$, which belongs to some bag $X_i$.

\paragraph*{Arc coverage}
Consider some arc in the spine.  This must join a triangle $t$ to an edge e that contains it.  Let $\delta$ be some tetrahedron containing $t$; then $\delta$ contains both $t$ and e, and so if $X_i$ is a bag containing $\delta$ then the corresponding bag $X'_i$ contains the chosen arc in the spine (joining $t$ with e).

\paragraph*{Coherence}
Here we treat edges and triangles separately. 

Let $t$ be some triangle in the simplicial complex.  We must show that the bags containing $t$ correspond to a connected subgraph of the underlying tree. If $t$ is a boundary triangle, then $t$ belongs to a unique tetrahedron $\delta$, and the bags $X'_i$ that contain $t$ correspond precisely to the bags $X_i$ that contain $\delta$.  Since the tree decomposition T satisfies the connectivity property, these bags correspond to a connected subgraph of the underlying tree. If $t$ is an internal triangle, then $t$ belongs to two tetrahedra $\delta_1$ and $\delta_2$, and the bags $X'_i$ that contain $t$ correspond to the bags $X_i$ that contain \emph{either} $\delta_1$ or $\delta_2$.  By the connectivity property of the original tree, the bags containing $\delta_1$ describe a connected subgraph of the tree, and so do the bags containing $\delta_2$. Furthermore, there is an arc in the dual graph from $\delta_1$ to $\delta_2$, and so some bag $X_i$ contains {\em both} $\delta_1$ and $\delta_2$.  Thus the union of these two connected subgraphs is another connected subgraph, and we have established the connectivity property for $t$.

Now let $e$ be some edge of the simplicial complex.  Again, we must show that the bags containing e correspond to a connected subgraph of the underlying tree.  This is simply an extension of the previous argument.  Suppose that e belongs to the tetrahedra $\delta_1,\ldots,\delta_m$ (ordered cyclically around e).  Then for each $\delta_j$, the bags $X_i$ that contain $\delta_j$ describe a connected subgraph of the underlying tree, and the bags $X'_j$ containing e describe the union of these subgraphs, which we need to show is again connected.  This follows because there is a sequence of arcs in the dual graph $(\delta_1,\delta_2)$, $(\delta_2,\delta_3)$ and so on; from the tree decomposition $T$ it follows that the subgraph for $\delta_1$ meets the subgraph for $\delta_2$, the subgraph for $\delta_2$ meets the subgraph for $\delta_3$, and so on.  Therefore the union of these subgraphs is itself connected.
\end{proof}

\end{document}